\documentclass[letterpaper, 10 pt, conference]{ieeeconf}  

\usepackage{amsfonts}
\usepackage{amssymb}
\usepackage[cmex10]{amsmath}
\usepackage{graphicx}
\usepackage{subfigure}
\usepackage{verbatim}
\usepackage{cite}
\usepackage{fancyhdr}
\usepackage{tabularx}
\usepackage{booktabs}
\usepackage[ruled]{algorithm2e}
\usepackage{multirow}
\usepackage{hyperref}
\usepackage{bm}
\usepackage{stfloats}
\usepackage{enumerate}
\usepackage{wasysym}
\usepackage{amsthm}
\usepackage{fancyhdr}
\usepackage{amsfonts}
\usepackage{mathrsfs}
\usepackage{pifont}
\usepackage{amssymb}
\usepackage{verbatim}
\usepackage{upgreek}
\usepackage{graphicx}
\usepackage{epstopdf}
\usepackage{subfigure}
\usepackage{color}
\usepackage{amsmath}
\usepackage{algorithmic}
\newcolumntype{C}{>{\Centering\arraybackslash}X} 

\IEEEoverridecommandlockouts                            
\title{\huge Safe Stabilization with Model Uncertainties:  \\
A Universal Formula with Gaussian Process Learning  }
\author{Ming Li and Zhiyong Sun} 

\begin{document}
\newtheorem*{Exam*}{Example}
\newtheorem*{Exam_Con*}{Example-Continued}
\newtheorem{Thm}{\textbf{Theorem}}
\newtheorem{Lem}{\textbf{Lemma}}
\newtheorem{Def}{\textbf{Definition}}
\newtheorem{Rem}{\textbf{Remark}}
\newtheorem{Exam}{\textbf{Example}}
\newtheorem{Sup}{\textbf{Assumption}}
\newtheorem{Cor}{\textbf{Corollary}}
\newtheorem{Asum}{\textbf{Assumption}}
\newtheorem{Expl}{\textbf{Explanation}}
\newtheorem{Prop}{\textbf{Proposition}}
\newtheorem{Rmk}{\textbf{Remark}}
\maketitle

\begin{abstract}
A combination of control Lyapunov functions (CLFs) and control barrier functions (CBFs) forms an efficient framework for addressing control challenges in safe stabilization. In our previous research, we developed an analytical control strategy, namely the universal formula, that incorporates CLF and CBF conditions for safe stabilization. However, successful implementation of this universal formula relies on an accurate model, as any mismatch between the model and the actual system can compromise stability and safety. In this paper, we propose a new universal formula that leverages Gaussian processes (GPs) learning to address safe stabilization in the presence of model uncertainty. By utilizing the results related to bounded learning errors, we  achieve a high probability of stability and safety guarantees with the proposed universal formula. Additionally, we introduce a probabilistic compatibility condition to evaluate conflicts between the modified CLF and CBF conditions with GP learning results. In cases where compatibility assumptions fail and control system limits are present, we propose a modified universal formula that relaxes stability constraints and a projection-based method accommodating control limits. We illustrate the effectiveness of our approach through a simulation of adaptive cruise control (ACC), highlighting its potential for practical applications in real-world scenarios.
\end{abstract}
\let\thefootnote\relax\footnotetext{This work was supported in part by a starting grant from Eindhoven Artificial Intelligence Systems Institute (EAISI), The Netherlands; in part by the National Natural Science Foundation of China (62073274). \emph{(Corresponding author: Zhiyong~Sun.)}

The authors are with the Department of Electrical Engineering, Eindhoven University of Technology, and also with the Eindhoven Artificial Intelligence Systems Institute, PO Box 513, Eindhoven 5600 MB, The Netherlands. 
{\tt\small \{ m.li3, z.sun \}@tue.nl}
}
\section{Introduction}
Safe stabilization, which refers to the task of guiding a system to a desired state within a predefined safety set, remains a persistent and significant challenge. Control Lyapunov functions (CLFs)~\cite{Global_Stabilizable} and control barrier functions (CBFs)~\cite{quadrotor_safe_stabilization2} are effective tools for addressing stabilization and safety-critical control problems, respectively. In recent years, they have been combined to resolve the safe stabilization problems in numerous applications, such as adaptive cruise control (ACC)~\cite{ACC_Application}, bipedal robot walking~\cite{bipedal_application}, and multi-robot coordination~\cite{multi_robot_coordination}.

In general, there are three major categories of solutions for combining CLFs and CBFs: i) construct a proper function that combines Lyapunov functions and barrier functions, namely control Lyapunov barrier function (CLBF)~\cite{CLBF}, and then a feedback control law is constructed using Sontag's universal formula~\cite{Sontag_fromula}, ii) formulate an optimization-based framework that utilizes the conditions of CLFs and CBFs with a quadratic program (QP)~\cite{CBF_Definition}, and iii)  design an analytical control law, named a universal formula, which provides a combination of the control law obtained from CLFs and CBFs, respectively. As shown in Fig~\ref{Summary}, an overview of the existing solutions for the combination of CLFs and CBFs is provided. 

\begin{figure*}[tp]
 \centering
    \makebox[0pt]{%
    \includegraphics[width=7in]{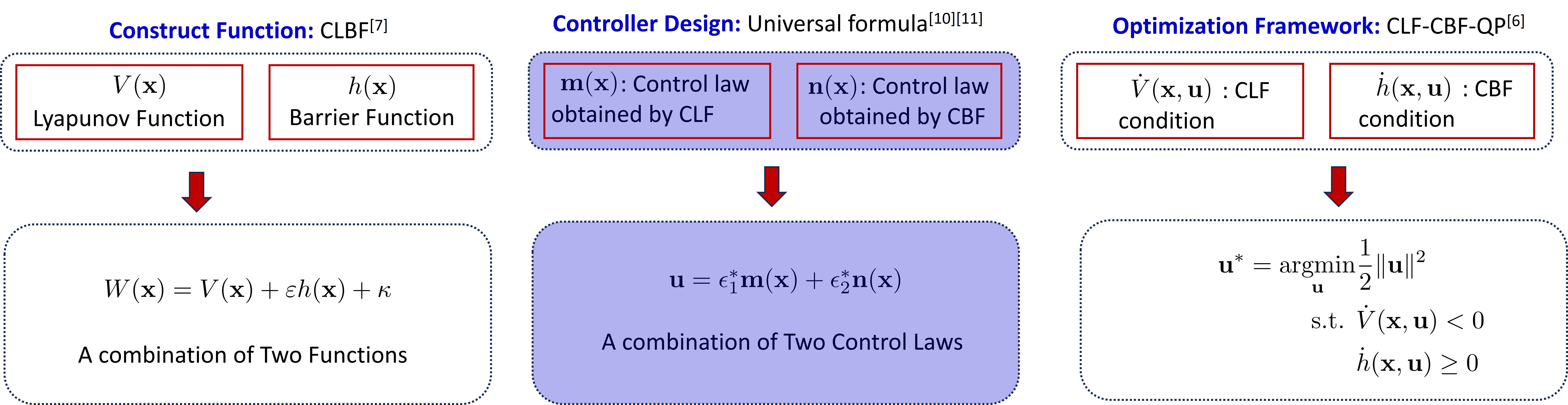}}
    \caption{An overview of the existing approaches that combine CLFs and CBFs to address safe stabilization problems.
   }
    \label{Summary}
\end{figure*}
The CLBF solution~\cite{CLBF}, while commended for its simplicity, encounters challenges in satisfying the prerequisites for the existence of CLBF~\cite{Property_on_CLBF}. Moreover, it may fail to simultaneously ensure stability and safety, even when the conditions of both CLFs and CBFs are possible to be satisfied~\cite {Comment_on_CLBF}. As for the optimization-based methods, a combination of CLFs and CBFs with a quadratic program (QP) formulation is proposed in~\cite{CBF_Definition}. Although they have demonstrated effectiveness, they are limited by the need to solve optimization problems, making it infeasible to perform real-time computations on computation resource-constrained platforms. 
 Moreover, because the control input is obtained through solving an optimization problem, it often leads to control laws that lack smoothness guarantees. 
In comparison, the controller design method offers several advantages. It is more straightforward since it directly manipulates control inputs, which aligns with the desired objectives. Consequently, it allows for direct parameter tuning within the control law to achieve specific properties, such as smoothness~\cite{Smooth_universal}. The universal formula is a special control design approach~\cite{universal_formula,Smooth_universal}, which provides an analytical formula concerning the conditions of CLFs and CBFs. As a result, it provides a fast computation speed and is suitable for platforms with constrained computation resources.

However, one major limitation shared by the universal formula, CLBF, and optimization-based solutions lies in their requirement for a precise system model. For safe stabilization, discrepancies between the actual system dynamics and the model can lead to a degradation in performance and result in unsafe behaviors. To resolve this challenge, various techniques have been introduced, including robust controllers~\cite{Comment_on_CLBF,ISSf}, observer-based design~\cite{Observer_CBF}, and machine learning-based methods~\cite{machine_learning_CBF,GP_CBF_Ref,CLBF_NN}, all aiming at mitigating the influence of model uncertainties. However, these endeavors have primarily concentrated on CLBF and optimization-based techniques. To the best of the authors' knowledge, no prior work has delved into \textit{controller design methods} for safe stabilization while accommodating model uncertainties.

Gaussian processes (GPs), a machine learning methodology renowned for its non-parametric probabilistic modeling framework and extensively used in designing controllers for systems with unknown dynamics, serve as the inspiration for our approach. In this paper, we integrate the previously proposed universal formula~\cite{universal_formula} with GPs to address the challenges of safe stabilization in the presence of model uncertainties. Utilizing a GP model to represent the system's unknown dynamics enables us to access both mean and variance functions. By incorporating the learning results into the conditions of CLFs and CBFs and leveraging the boundedness of the learning error, we  provide high probability safety assurances, even when the model is inaccurate. Meantime, we present the compatibility conditions, which assess conflicts between CLF and CBF conditions, to evaluate whether a safe stabilization problem can be addressed with our universal formula with GPs. For the cases where compatibility assumption does not hold and control system limits exist, we introduce  a modified universal formula that relaxes stability constraints and a projection-based method that takes into account control limits. Finally, we illustrate the effectiveness of our approach through a simulation of adaptive cruise control (ACC), which showcases its potential for practical applications in real-world scenarios.
\section{Preliminaries and Problem Statement}
Consider a dynamic system that has a control-affine structure~\cite{Dynamical_model}.
	\begin{equation}\label{Affine_Control_System}
	    \dot{\mathbf{x}}=\mathbf{f}(\mathbf{x})+\mathbf{g}(\mathbf{x})\mathbf{u},\quad \mathbf{x}(0)=\mathbf{x}_{0},
	\end{equation}
where $\mathbf{x}\in\mathbb{R}^{n}$, $\mathbf{u}\in\mathbb{R}^{m}$, and the vector fileds $\mathbf{f}:\mathbb{R}^{n}\rightarrow\mathbb{R}^{n}$ and $\mathbf{g}:\mathbb{R}^{n}\rightarrow\mathbb{R}^{n\times m}$ are  assumed to be locally Lipschitz continuous and $\mathbf{f}(\mathbf{0})=\mathbf{0}$.


\subsection{CLF and Universal Formulas}
Firstly, we provide a formal definition of a CLF in Definition~\ref{CLF}. Compared to the definition in~\cite{Sontag_fromula}, the standard CLF condition is tightened to demonstrate its relationship with the universal formula.
\begin{Def}\label{CLF}
(CLF) A continuously differentiable, positive definite, and radially unbounded function $V: \mathbb{R}^{n}\rightarrow\mathbb{R}_{+}$ is a CLF for system \eqref{Affine_Control_System} if there exists a control $\mathbf{u}\in\mathbb{R}^{m}$ satisfying
\begin{equation}\label{CLF_Condition}
\begin{aligned}
& a(\mathbf{x})+\mathbf{b}(\mathbf{x}) \mathbf{u}\leq -\kappa\zeta(\mathbf{x})
\end{aligned}
\end{equation}
for all $\mathbf{x}\in\mathbb{R}^{n}\backslash\{ \mathbf{0}\}$, where $a(\mathbf{x})=L_{\mathbf{f}} V(\mathbf{x})+\lambda V(\mathbf{x})$, $\mathbf{b}(\mathbf{x})=L_{\mathbf{g}} V(\mathbf{x})$, $\kappa\geq 0$, and $\zeta(\mathbf{x})$ is a positive definite function. $L_{\mathbf{f}}$ and $L_{\mathbf{g}}$ denote the Lie derivatives along $\mathbf{f}$ and $\mathbf{g}$, respectively.
\end{Def}
Referring to the condition~\eqref{CLF_Condition}, we can always guarantee that $a(\mathbf{x})+\mathbf{b}(\mathbf{x}) \mathbf{u} \leq  0$ for all $\mathbf{x}\in\mathbb{R}^{n}\backslash\{ \mathbf{0}\}$, which is the standard CLF condition in~\cite{Sontag_fromula}. Therefore, we claim that the condition \eqref{CLF_Condition} serves as a sufficient condition of the standard CLF condition, and it degrades to the standard one when $\kappa=0$. 
Thus, the goal is to obtain a state feedback control law $\mathbf{u}: \mathbb{R}^{n}\rightarrow\mathbb{R}^{m}$ that satisfies the condition~\eqref{CLF_Condition} for any $\mathbf{x}\in\mathbb{R}^{n}\backslash{ \{\mathbf{0}\}}$. Toward this objective, Sontag's universal formula~\cite{Sontag_fromula} can be employed.
\begin{equation}\label{CLF_Sontag_Law}
\begin{split}
\mathbf{u}_{\mathrm{Uni-CLF}}^{\star}(\mathbf{x})=\left\{\begin{array}{cc}
\mathbf{m}_{\mathrm{Uni-CLF}}^{\star}(\mathbf{x}), & \mathbf{b}(\mathbf{x}) \neq \mathbf{0} \\
\mathbf{0}, & \mathbf{b}(\mathbf{x})=\mathbf{0}
\end{array}\right.
\end{split}
\end{equation}
where  $\mathbf{m}_{\mathrm{Uni-CLF}}^{\star}(\mathbf{x})=-\frac{a(\mathbf{x})+\kappa\zeta(\mathbf{x})}{\mathbf{b}(\mathbf{x}) \mathbf{b}(\mathbf{x})^{\top}} \mathbf{b}(\mathbf{x})^{\top}$, $\phi(\mathbf{x})$ is a positive semi-definite function, and we select $\zeta(\mathbf{x})=\sqrt{a^{2}(\mathbf{x})+\phi(\mathbf{x})\|\mathbf{b}(\mathbf{x}) \|^{4}}$. One can observe that the universal formula provided in~\eqref{CLF_Sontag_Law} ensures the satisfaction of~\eqref{CLF_Condition}, which means that the closed-loop stability is guaranteed. Additionally,   in contrast to the universal formula presented in~\cite{Sontag_fromula}, our introduced parameter $\kappa\geq 0$ serves as a key addition to the universal formula. This parameter introduces the flexibility to control the speed of stabilization.

\subsection{CBF and Universal Formulas}
Consider a closed convex set $\mathcal{C}\subset\mathbb{R}^{n}$ as the $0$-superlevel set of a  continuously differentiable function $h:\mathbb{R}^{n}\rightarrow\mathbb{R}$, which is defined as
\begin{equation}\label{Invariant_Set}
		\begin{aligned}
			\mathcal{C} & \triangleq\left\{\mathbf{x}\in\mathbb{R}^{n}: h(\mathbf{x}) \geq 0\right\} \\
			\partial \mathcal{C}, & \triangleq\left\{\mathbf{x}\in\mathbb{R}^{n}: h(\mathbf{x})=0\right\}, \\
			\operatorname{Int}(\mathcal{C}) & \triangleq\left\{\mathbf{x}\in\mathbb{R}^{n}: h(\mathbf{x})>0\right\},
		\end{aligned}
\end{equation}
where we assume that $\mathcal{C}$ is nonempty and has no isolated points, that is, $\operatorname{Int}(\mathcal{C}) \neq \emptyset$ and $\overline{\operatorname{Int}(\mathcal{C})}=\mathcal{C}$.
Similar to the CLF definition, we establish a sufficient condition for the standard CBF condition, as defined in~\cite{Zeroing_CBF}, to demonstrate that the universal formula follows a more conservative condition than the standard CBF.
\begin{Def}\label{CBF_Def}
(CBF) Let $\mathcal{C}\subset\mathbb{R}^{n}$ be the $0$-superlevel set of a continuously differentiable function $h:\mathbb{R}^{n}\rightarrow\mathbb{R}$ which is defined by \eqref{Invariant_Set}. Then $h$ is a CBF for \eqref{Affine_Control_System} if, for all $\mathbf{x}\in\mathcal{C}$, there exists a control $\mathbf{u}\in\mathbb{R}^{m}$ satisfying
		\begin{equation}\label{CBF_Condition}
    \begin{split}
        &c(\mathbf{x})+\mathbf{d}(\mathbf{x})\mathbf{u}\geq \rho\Gamma(\mathbf{x})
    \end{split}
\end{equation}
where $c(\mathbf{x})=L_{\mathbf{f}} h(\mathbf{x})+\beta{h(\mathbf{x})}$, $\beta\geq 0$, $\mathbf{d}(\mathbf{x})=L_{\mathbf{g}} h(\mathbf{x})$, $\rho\geq 0$, and $\Gamma(\mathbf{x})$ is a positive definite function.
\end{Def}
The universal control law for CBF is given as 
\begin{equation}\label{CBF_Universal_Law}
\begin{split}
\mathbf{u}_{\mathrm{Uni-CBF}}^{\star}(\mathbf{x})=\left\{\begin{array}{cc}
\mathbf{n}_{\mathrm{Uni-CBF}}^{\star}(\mathbf{x}), & \mathbf{d}(\mathbf{x})\neq\mathbf{0},\\ \mathbf{0},&\mathbf{d}(\mathbf{x})=\mathbf{0},
\end{array}\right.
\end{split}
\end{equation}
where $\mathbf{n}_{\mathrm{Uni-CBF}}^{\star}(\mathbf{x})=\frac{\rho\Gamma(\mathbf{x})-c(\mathbf{x})}{\mathbf{d}(\mathbf{x}) \mathbf{d}(\mathbf{x})^{\top}}\mathbf{d}(\mathbf{x})^{\top}$, $\varphi(\mathbf{x})$ is a positive semi-definite function, and we choose $\Gamma(\mathbf{x})=\sqrt{c^{2}(\mathbf{x})+\varphi(\mathbf{x})\|\mathbf{d}(\mathbf{x}) \|^{4}}$. Note that the parameter $\rho\geq 0$ is employed to regulate the level of conservatism in the safety guarantees.
\subsection{Universal Formula for Safe Stabilization}
\begin{Def}
(Compatibility of CLF and CBF) The CLF $V(\mathbf{x})$ and CBF $h(\mathbf{x})$ are compatible for the dyanmical model~\eqref{Affine_Control_System} if there exists a control input $\mathbf{u}\in\mathbb{R}^{m}$ satisfying~\eqref{CLF_Condition} and~\eqref{CBF_Condition} simultaneously. If $V(\mathbf{x})$ and  $h(\mathbf{x})$ are compatible at every point of $\mathbf{x}\in\mathbb{R}^{n}$, the functions $V:\mathbb{R}^{n}\rightarrow\mathbb{R}_{>0}$ and $h:\mathbb{R}^{n}\rightarrow\mathbb{R}_{>0}$ are compatible.
\end{Def}
As illustrated in Fig.~\ref{Incompatible}, one collision avoidance example and the graphical interpretation of compatibility are provided. Note that $\Gamma_{c}: =\{\mathbf{x}|V(\mathbf{x})\leq c\}, c>0$, represented by the black ellipsoid curves, denotes the domain of attraction. In certain regions, such as those containing the depicted blue curves, $V(\mathbf{x})$ and $h(\mathbf{x})$ are compatible across all states. However, the yellow curve becomes stuck at the edge of the safe region $\mathcal{C}$. It indicates that, at the boundary point corresponding to a state $\mathbf{x}$, the CLF $V(\mathbf{x})$ and CBF $h(\mathbf{x})$ are incompatible. 
\begin{figure}[tp]
 \centering
    \makebox[0pt]{%
    \includegraphics[width=1.8in]{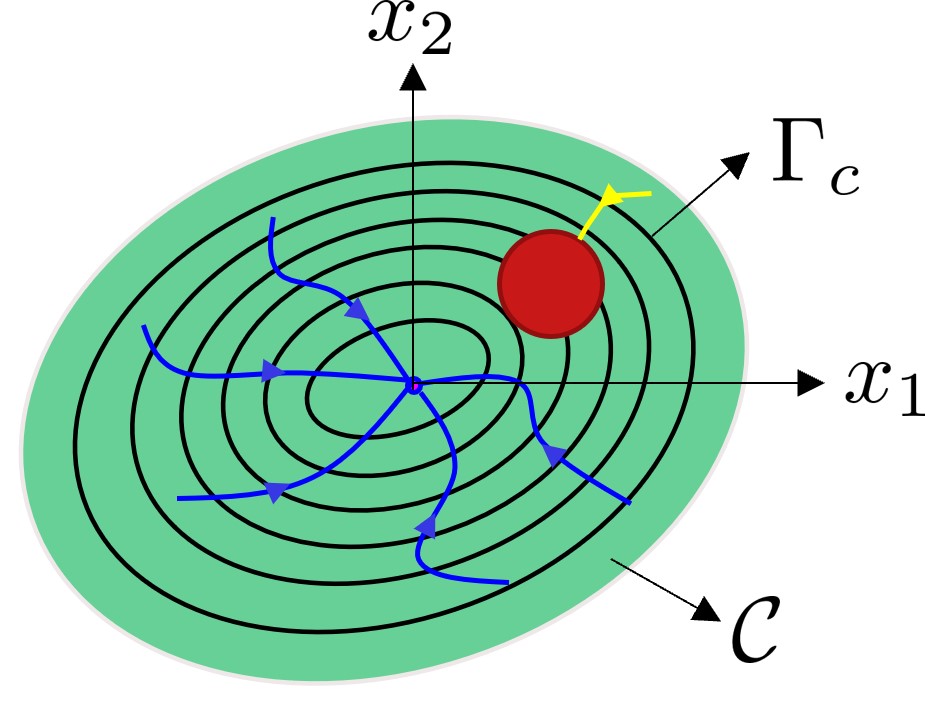}}
    \caption{A collision avoidance example and the graphical interpretation of compatibility: the green area represents the safe set $\mathcal{C}$, the shaded red area indicates the unsafe region, the black ellipsoidal curve denotes the domain of attraction $\Gamma_{c}$, the blue trajectories depict the convergence to the origin from various initial states, and the yellow curve serves to illustrate the example of incompatibility occurring at certain states.
   }
    \label{Incompatible}
\end{figure}
\begin{Lem}\label{Compatible_Lemma}
(\cite{universal_formula}) Assume that both $\mathbf{b}(\mathbf{x})$ and $\mathbf{d}(\mathbf{x})$ are nonzero vectors. The CLF $V(\mathbf{x})$ and CBF $h(\mathbf{x})$ for the model~\eqref{Affine_Control_System} are compatible if and only if one of the following conditions is satisfied. 
\begin{equation}\label{Compatibility_conditions}
\begin{split}
\begin{cases}
     &\frac{\|\mathbf{b}(\mathbf{x})\mathbf{d}(\mathbf{x})^{\top}\|}{\|\mathbf{b}(\mathbf{x})\|\|\mathbf{d}(\mathbf{x})\|}\neq 1,\\
     &\frac{\|\mathbf{b}(\mathbf{x})\mathbf{d}(\mathbf{x})^{\top}\|}{\|\mathbf{b}(\mathbf{x})\|\|\mathbf{d}(\mathbf{x})\|}= 1\,\,\, \text{and}\,\,\, v(\mathbf{x})\geq 0, \\
     &\frac{\|\mathbf{b}(\mathbf{x})\mathbf{d}(\mathbf{x})^{\top}\|}{\|\mathbf{b}(\mathbf{x})\|\|\mathbf{d}(\mathbf{x})\|}= 1\,\,\, \text{and}\,\,\, w(\mathbf{x})\geq 0,
     \end{cases}
\end{split}
\end{equation}
where $w(\mathbf{x})=(a(\mathbf{x})+\kappa \zeta(\mathbf{x})) \mathbf{d}(\mathbf{x})^{\top} \mathbf{d}(\mathbf{x})-(c(\mathbf{x})-\rho \Gamma(\mathbf{x})) \mathbf{b}(\mathbf{x})^{\top} \mathbf{d}(\mathbf{x})$ and $ v(\mathbf{x})=(a(\mathbf{x})+\kappa \zeta(\mathbf{x})) \mathbf{d}(\mathbf{x})^{\top} \mathbf{b}(\mathbf{x})-(\mathbf{c}(\mathbf{x})-\rho \Gamma(\mathbf{x})) \mathbf{b}(\mathbf{x})^{\top} \mathbf{b}(\mathbf{x})$.
\end{Lem}
The two conditions provided in \eqref{CLF_Condition} and \eqref{CBF_Condition} correspond to two different hyperplanes for each $\mathbf{x}\in\mathbb{R}^{n}$. The three cases described in Lemma~\ref{Compatible_Lemma} are as follows: i) when the two hyperplanes are non-parallel, ii) when the hyperplanes are parallel but the one defined by the CLF is positioned above the one defined by the CBF, and iii) a condition similar to ii) but the other way around.
\begin{Thm}\label{Compatible_Qp}
(\cite{universal_formula}) Assume that the CLF $V(\mathbf{x})$ and CBF $h(\mathbf{x})$ for the model~\eqref{Affine_Control_System} are compatible. The generalized universal formula~\eqref{QP_Control_Law_relaxed} ensures both the safety and stability of~\eqref{Affine_Control_System} simultaneously.
\begin{equation}\label{QP_Control_Law_relaxed}
    \begin{split}
      \mathbf{u}_{\mathrm{Uni}}^{\star}(\mathbf{x})=\left\{\begin{array}{ll}
\mathbf{m}_{\mathrm{Uni-CLF}}^{\star}(\mathbf{x}),&\quad\mathbf{x}\in\mathcal{P}_{1}\\
\mathbf{n}_{\mathrm{Uni-CBF}}^{\star}(\mathbf{x}),&\quad\mathbf{x}\in\mathcal{P}_{2}\\
\mathbf{p}_{\mathrm{Uni}}^{\star}(\mathbf{x}),&\quad\mathbf{x}\in\mathcal{P}_{3}\\
\mathbf{0},&\quad\mathbf{x}\in\mathcal{P}_{4},
\end{array}\right.  
    \end{split}
\end{equation}
where $\mathbf{m}_{\mathrm{Uni-CLF}}^{\star}(\mathbf{x})=-\frac{a(\mathbf{x})+\kappa\zeta(\mathbf{x})}{\mathbf{b}(\mathbf{x})\mathbf{b}(\mathbf{x})^{\top}} \mathbf{b}(\mathbf{x})^{\top}$, $\mathbf{n}_{\mathrm{Uni-CBF}}^{\star}(\mathbf{x})= \frac{\rho\Gamma(\mathbf{x})- c(\mathbf{x})}{\mathbf{d}(\mathbf{x})\mathbf{d}(\mathbf{x})^{\top}} \mathbf{d}(\mathbf{x})^{\top}$,  $\mathbf{p}_{\mathrm{Uni}}^{\star}=\epsilon_{1}\mathbf{m}_{\mathrm{CLF}}^{\star}(\mathbf{x})
+\epsilon_{2}\mathbf{n}_{\mathrm{CBF}}^{\star}(\mathbf{x})$, $\lambda_{2}=-\lambda_{1}\frac{\mathbf{b}(\mathbf{x})\mathbf{b}(\mathbf{x})^{\top}}{a(\mathbf{x})+\kappa\zeta(\mathbf{x})}$, $\epsilon_{3}=\lambda_{2}\frac{\mathbf{d}(\mathbf{x})\mathbf{d}(\mathbf{x})^{\top}}{\rho\Gamma(\mathbf{x})-c(\mathbf{x})}$,
\begin{equation*}\label{bar_lambda_compute}
    \begin{split}
        \begin{bmatrix}
\lambda_{1}\\
\lambda_{2}
\end{bmatrix}={\begin{bmatrix}
\mathbf{b}(\mathbf{x})\mathbf{b}(\mathbf{x})^{\top} & -\mathbf{b}(\mathbf{x})\mathbf{d}(\mathbf{x})^{\top}\\
-\mathbf{d}(\mathbf{x})\mathbf{b}(\mathbf{x})^{\top} & \mathbf{d}(\mathbf{x})\mathbf{d}(\mathbf{x})^{\top}
\end{bmatrix}}^{-1}\begin{bmatrix}
a(\mathbf{x})+\kappa\zeta(\mathbf{x})\\
\rho\Gamma(\mathbf{x})-c(\mathbf{x})
\end{bmatrix},
    \end{split}
\end{equation*}
\begin{equation*}\label{Domain_of_sets_relaxed}
\begin{aligned}
&\mathcal{P}_{1}=\{\mathbf{x}\in\mathbb{R}^{n}|a(\mathbf{x})+\kappa\zeta(\mathbf{x})\geq 0,v(\mathbf{x})< 0\},\\
&\mathcal{P}_{2}=\{\mathbf{x}\in\mathbb{R}^{n}|c(\mathbf{x})-\rho\Gamma(\mathbf{x})\leq 0,w(\mathbf{x})< 0\},\\
&\mathcal{P}_{3}=\{\mathbf{x}\in\mathbb{R}^{n}|w(\mathbf{x})\geq 0,v(\mathbf{x})\geq 0, \\
&\mathbf{b}(\mathbf{x})^{\top} \mathbf{b}(\mathbf{x}) \mathbf{d}(\mathbf{x})^{\top} \mathbf{d}(\mathbf{x})-\mathbf{b}(\mathbf{x})^{\top} \mathbf{d}(\mathbf{x}) \mathbf{b}(\mathbf{x})^{\top} \mathbf{d}(\mathbf{x}) \neq 0\},\\
&\mathcal{P}_{4}=\{\mathbf{x}\in\mathbb{R}^{n}|a(\mathbf{x})+\kappa\zeta(\mathbf{x})< 0,c(\mathbf{x})-\rho\Gamma(\mathbf{x})> 0\}.
\end{aligned}
\end{equation*}
\end{Thm}
\begin{Rmk}
As we see in~\eqref{QP_Control_Law_relaxed}, the universal formula is a piece-wise function, where certain regions correspond to a specific control law. Generally, the piece-wise function can be rewritten in a more compact form as a combination of $\mathbf{m}_{\mathrm{Uni-CLF}}^{\star}(\mathbf{x})$ and $\mathbf{n}_{\mathrm{Uni-CLF}}^{\star}(\mathbf{x})$, i.e., $\mathbf{u}_{\mathrm{Uni}}^{\star}(\mathbf{x})=\epsilon_{1}^{*}\mathbf{m}_{\mathrm{Uni-CLF}}^{\star}(\mathbf{x})+\epsilon_{2}^{*}\mathbf{n}_{\mathrm{Uni-CLF}}^{\star}(\mathbf{x})$, where 
\begin{equation*}\label{Parameters}
    \begin{split}
     \left\{\begin{array}{ll}
      \epsilon_{1}^{*}=1,\,\, \epsilon_{2}^{*}=0&\quad\mathbf{x}\in\mathcal{P}_{1}\\
\epsilon_{1}^{*}=0, \,\,\epsilon_{2}^{*}=1&\quad\mathbf{x}\in\mathcal{P}_{2}\\
\epsilon_{1}^{*}=\epsilon_{1}, \epsilon_{2}^{*}=\epsilon_{2}&\quad\mathbf{x}\in\mathcal{P}_{3}\\
\epsilon_{1}^{*}=0, \,\,\epsilon_{2}^{*}=0,&\quad\mathbf{x}\in\mathcal{P}_{4}.
\end{array}\right.  
    \end{split}
\end{equation*}
However, this is not the only approach to derive a universal formula. As demonstrated in~\cite{Smooth_universal}, an alternative method involves employing a smooth function to define and determine the parameters $\epsilon_{1}^{*}$ and $\epsilon_{2}^{*}$. Many more universal formulas can be constructed by carefully selecting the parameters $\epsilon_{1}^{*}$ and $\epsilon_{2}^{*}$, which will result in different closed-loop control properties.
\end{Rmk}
\subsection{Gaussian Process Regression}
Consider a dynamic model as follows.
\begin{equation}\label{Affine_Control_System_Unknown}
\dot{\mathbf{x}}=\mathbf{f}(\mathbf{x})+\mathbf{g}(\mathbf{x})\mathbf{u}+\bm{\omega}(\mathbf{x}),
\end{equation}
where $\bm{\omega}(\mathbf{x})$ is Lipshitz continuous and denotes the unknown, but deterministic, state-dependent disturbance capturing unmodeled dynamics. 
\begin{Rmk}
    There are many different choices for defining a dynamic model with unknown dynamics, and the specific settings employed vary based on different considerations. For instance, in~\cite{Fully_Unknown_Dynamics}, the authors are interested in the model~\eqref{Affine_Control_System}, where both $\mathbf{f}(\mathbf{x})$ and $\mathbf{g}(\mathbf{x})$ are unknown. However, some papers focus on scenarios where only $\mathbf{f}(\mathbf{x})$ is fully unknown while $\mathbf{g}(\mathbf{x})$ is available, as seen in~\cite{Dynamical_model}. In other cases, such as~\cite{GP_solution2}, a nominal model of the form $\dot{\mathbf{x}}=\tilde{\mathbf{f}}(\mathbf{x})+\tilde{\mathbf{g}}(\mathbf{x})\mathbf{u}$ is provided, where $\tilde{\mathbf{f}}(\mathbf{x})\neq\mathbf{f}(\mathbf{x})$ and $\tilde{\mathbf{g}}(\mathbf{x})\neq\mathbf{g}(\mathbf{x})$. In contrast, this paper focuses on the system defined by \eqref{Affine_Control_System_Unknown}, where $\mathbf{f}(\mathbf{x})$ and $\mathbf{g}(\mathbf{x})$ are perfectly known while an unknown term $\bm{\omega}(\mathbf{x})$ is introduced. An illustrative example of such a consideration can be found in Section~\ref{Application_study}, where $\bm{\omega}(\mathbf{x})$ indicates the mismatch in model parameters.
\end{Rmk}

We can learn the function $\bm{\omega}(\mathbf{x})$ from past data by formulating a supervised learning problem.
GP is a non-parametric regression method, where the goal is to find an approximation of unknown nonlinear dynamics $\bm{\omega}: \mathbb{R}^{n}\rightarrow\mathbb{R}^{n}$ with $n$ independent GPs given as
\begin{equation}
    \bm{\omega}(\mathbf{x})=\left\{\begin{array}{c}
\bm{\omega}_1(\mathbf{x}) \sim \mathcal{GP}\left(\mathbf{m}_{1}(\mathbf{x}), \mathbf{k}_1\left(\mathbf{x}, \mathbf{x}^{\prime}\right)\right), \\
\vdots \\
\bm{\omega}_n(\mathbf{x}) \sim \mathcal{GP}\left(\mathbf{m}_{n}(\mathbf{x}), \mathbf{k}_n\left(\mathbf{x}, \mathbf{x}^{\prime}\right)\right),
\end{array}\right.
\end{equation}
where $\mathbf{m}_{i}(\mathbf{x}), i=1,\cdots n$ is a mean function, $\mathbf{k}_i\left(\mathbf{x}, \mathbf{x}^{\prime}\right)$ is a covariance function (a.k.a., a symmetric positive definite function referred to as kernel). The most commonly used kernels include the linear, squared exponential, and \text{Matèrr} kernels. In a function space view, such $k(\mathbf{x},\mathbf{x}^{\prime})$ prescribes a specific class of Hilbert space called reproducing kernel Hilbert space (RKHS,~\cite{Kernel_Functions}), denoted as $\mathcal{H}_{k}(\mathcal{X})$, where $\mathcal{X}$ is the domain of the kernel function, a connected subset of $\mathbb{R}^{n}$. Moreover, the RKHS norm $\|h\|_{k}:=\sqrt{\langle h,h\rangle_{\mathcal{H}_{k}(\mathcal{X})}}$ is a measure of its smoothness with respect to the kernel function.
\begin{Asum}\label{Data_Collection}
(\cite{Armin_paper})
We assume that we have access to the state  $\mathbf{x}$  of~\eqref{Affine_Control_System_Unknown}, and thus the unmodeled dynamics can be measured. Therefore, the data set $\mathbb{D}$, which is composed of $Q$ pairs of $\mathbf{x}^{(q)}$ and $\mathbf{y}^{(q)}$, is available  
\begin{equation}\label{Data_set}
\mathbb{D}=\left\{\mathbf{x}^{(q)}, \mathbf{y}^{(q)}=\bm{\omega}\left(\mathbf{x}^{(q)}\right)+\bm{\varepsilon}^{(q)}\right\}_{q=1}^Q.
\end{equation}
Herein, $\mathbf{x}^{(q)}, q=1,\cdots N$ is a sample of noiseless measurement of the state $\mathbf{x}$, $\bm{\varepsilon}^{(q)}\sim\mathcal{N}(\mathbf{0}_{n},\sigma_{\bm{\varepsilon}}^{2}\mathbf{I}_{n})$ is a sample of an additive noise, and $\mathbf{y}^{(q)}$ is the measurement of $\bm{\omega}(\mathbf{x}^{(q)})$.
\end{Asum}
Assuming that Assumption~\ref{Data_Collection} holds, the posterior distribution for $\bm{\omega}_{i}(\mathbf{x})$ at a query point $\mathbf{x}_{*}$ is calculated as a Gaussian distribution $\mathcal{N}(\mathbf{m}_{i}(\mathbf{x}),\bm{\sigma}_{i}^{2}(\mathbf{x}))$ with the following mean and covariance:
\begin{equation}\label{mean_function}
    \mathbf{m}_{i}(\mathbf{x})=\mathbf{k}_{i,*}^{\top}(\mathbf{K}_{i}+\sigma_{\bm{\varepsilon}}^{2}\mathbf{I}_{N})^{-1}\mathbf{y}_{i},
\end{equation}
\begin{equation}\label{covariance_function}
    \bm{\sigma}_{i}^{2}(\mathbf{x})=\mathbf{k}_{i}(\mathbf{x},\mathbf{x})-\mathbf{k}_{i,*}^{\top}(\mathbf{K}_{i}+\sigma_{\bm{\varepsilon}}^{2}\mathbf{I}_{N})^{-1}\mathbf{k}_{i,*}.
\end{equation}
where $\mathbf{k}_{i,*}=\left[k_{i}(\mathbf{x}^{(1)},\mathbf{x}_{*})\cdots,k_{i}(\mathbf{x}^{(Q)},\mathbf{x}_{*})\right]^{\top}\in\mathbb{R}^{N}$, 
\begin{equation*}
    \mathbf{K}_{i}=\left[\begin{array}{ccc}
\mathbf{k}_i\left(\mathbf{x}^{(1)}, \mathbf{x}^{(1)}\right) & \cdots & \mathbf{k}_i\left(\mathbf{x}^{(1)}, \mathbf{x}^{(N)}\right) \\
\vdots & \ddots & \vdots \\
\mathbf{k}_i\left(\mathbf{x}^{(N)}, \mathbf{x}^{(1)}\right) & \cdots & \mathbf{k}_i\left(\mathbf{x}^{(N)}, \mathbf{x}^{(N)}\right)
\end{array}\right] \in \mathbb{R}^{N \times N} .
\end{equation*}
The approximation of overall $\bm{\omega}(\mathbf{x})$ can be obtained by concatenating $\mathbf{m}_{i}(\mathbf{x})$ and $\bm{\sigma}_{i}^{2}(\mathbf{x})$ as follows.
\begin{equation}\label{mean_and_covariance_functions}
\begin{split}
    \mathbf{m}(\mathbf{x})&=\left[\mathbf{m}_{1}(\mathbf{x}),\cdots,\mathbf{m}_{N}(\mathbf{x})\right]^{\top}, \\
    \bm{\sigma}^{2}(\mathbf{x})&=\left[\bm{\sigma}_{1}^{2}(\mathbf{x}),\cdots,\bm{\sigma}_{n}^{2}(\mathbf{x})\right].
\end{split}
\end{equation}
With a prescribed probability $\delta$, the difference between the actual value of $\bm{\omega}_{i}(\mathbf{x})$ and the estimated mean $\mathbf{m}_{i}$ can be  bounded, as indicated in the following lemma.
\begin{Lem}\label{Bound_Learning_Error}
    Assume the function of the unmodeled dynamics $\bm{\omega}(\mathbf{x})$ has a finite RKHS norm with respect to the kernel $k$, i.e., $\|\bm{\omega}_{i}\|_{k}\leq\infty, \forall i=1,\cdots, n$. For any $\mathbf{x}\in\mathbb{R}^{n}$, there holds
\begin{small}
\begin{equation}\label{Bound_Condition}
    \begin{split}
    &\mathbb{P}\left\{\bigcap_{i=1}^n\mathbf{m}_{i}(\mathbf{x})-\beta_{i}\bm{\sigma}_{i}(\mathbf{x})\leq\omega_{i}(\mathbf{x})\leq\mathbf{m}_{i}(\mathbf{x})+\beta_{i}\bm{\sigma}_{i}(\mathbf{x})\right\}\\
    &\qquad\qquad\qquad\qquad\qquad\qquad\qquad\qquad\qquad\geq (1-\delta)^{n},
    \end{split}
\end{equation}
\end{small}

\noindent
with $\delta\in(0,1)$ and $\beta_{i}=\sqrt{2\|\bm{\omega}_{i}\|_{k}^{2}+300\gamma_{i}\ln^{3}\left(\frac{Q+1}{\delta}\right)}$, where $\gamma_{i}\in\mathbb{R}$ represents the information gain, and its selection can follow the guidelines outlined in~\cite{Gamma_Selection}.
\end{Lem}
\subsection{Problem Statement}
The CLF and CBF constraints for the dynamical system~\eqref{Affine_Control_System_Unknown} are given as follows
\begin{equation}\label{CLF_CBF_Condition}
\begin{split}
&a(\mathbf{x})+\mathbf{b}(\mathbf{x})\mathbf{u}+\frac{\partial V(\mathbf{x})}{\partial \mathbf{x}}\bm{\omega}(\mathbf{x})\leq -\kappa\zeta(\mathbf{x}),\\
&c(\mathbf{x})+\mathbf{d}(\mathbf{x})\mathbf{u}+\frac{\partial h(\mathbf{x})}{\partial \mathbf{x}}\bm{\omega}(\mathbf{x})\geq\rho\Gamma(\mathbf{x}).
\end{split}
\end{equation}
The goal of this paper is to find a universal formula as in~\eqref{QP_Control_Law_relaxed} that satisfies both of the CLF and CBF constraints in~\eqref{CLF_CBF_Condition}. However, due to the existence of unknown dynamics in~\eqref{Affine_Control_System_Unknown}, i.e., $\bm{\omega}(\mathbf{x})$, the derivation of a universal formula is unattainable. To address this problem, we propose to use GP learning to approximate $\bm{\omega}(\mathbf{x})$ and combine these learned results with a universal formula. Using the dataset given in~\eqref{Data_set} and employing GP regression methods to learn the unmodeled dynamics $\bm{\omega}(\mathbf{x})$ in~\eqref{Affine_Control_System_Unknown}, we can approximate $\bm{\omega}(\mathbf{x})$ with the mean and covariance functions provided in~\eqref{mean_and_covariance_functions}. By considering the learning errors, we aim to ensure a high probability of achieving safety and stability. Formally, the problem can be stated as follows.

\textit{Problem:}  In this paper, the research problem is to find a universal formula that incorporates the learned results (mean functions and covariance functions) to ensure that the constraints given in ~\eqref{CLF_CBF_Condition} are satisfied with a high probability, as shown in Lemma~\ref{Bound_Learning_Error}.
\begin{Rmk}
    Instead of employing GP regression to estimate the model uncertainty term $\bm{\omega}(\mathbf{x})$, an alternative approach can also be employed, as demonstrated in~\cite{machine_learning_CBF,GP_solution2}. This approach aims to learn the terms $\frac{\partial V(\mathbf{x})}{\partial \mathbf{x}}\bm{\omega}(\mathbf{x})$ and $\frac{\partial h(\mathbf{x})}{\partial \mathbf{x}}\bm{\omega}(\mathbf{x})$ rather than $\bm{\omega}(\mathbf{x})$. One notable advantage of the latter approach is that the dimensions of $\frac{\partial V(\mathbf{x})}{\partial \mathbf{x}}\bm{\omega}(\mathbf{x})$ and $\frac{\partial V(\mathbf{x})}{\partial \mathbf{x}}\bm{\omega}(\mathbf{x})$ are both one, which is considerably lower than the dimension of the uncertainty term $\bm{\omega}(\mathbf{x})$. This reduction in dimensionality results in a simplified GP learning process, reducing its computational complexity. In contrast, the advantage of directly learning $\bm{\omega}(\mathbf{x})$ lies in the fact that $\bm{\omega}(\mathbf{x})$ maintains the information of the model. Learning $\bm{\omega}(\mathbf{x})$ provides an interpretable representation, making it easier to understand the behavior of the model uncertainty and its influence on stability and safety.
\end{Rmk}
\section{Universal Formula with Gaussian Processes}
Based on Lemma~\ref{Bound_Learning_Error}, it becomes evident that the error in estimating the unknown dynamics $\bm{\omega}(\mathbf{x})$ via GP learning is bounded with a probability of $(1-\delta)^{n}$. Under the conditions outlined in~\eqref{CLF_CBF_Condition}, the following lemma is provided.
\begin{Lem}\label{GP_Sufficient}
Suppose the following conditions hold true for all $\mathbf{x}\in\mathcal{D}\subset\mathbb{R}^{n}$
\begin{equation}\label{CLF_CBF_GP_Condition}
\begin{split}
&\tilde{a}(\mathbf{x})+\mathbf{b}(\mathbf{x})\mathbf{u}\leq -\kappa\zeta(\mathbf{x}),\quad\tilde{c}(\mathbf{x})+\mathbf{d}(\mathbf{x})\mathbf{u}\geq \rho\Gamma(\mathbf{x}),
\end{split}
\end{equation}
where $\tilde{a}(\mathbf{x})=a(\mathbf{x})+\Delta_{V}$, $\tilde{c}(\mathbf{x})=c(\mathbf{x})+\Delta_{h}$, $\Delta_{V}=\frac{\partial V(\mathbf{x})}{\partial \mathbf{x}}\mathbf{m}(\mathbf{x})+\beta\left\|\frac{\partial V(\mathbf{x})}{\partial \mathbf{x}}\right\|\bm{\sigma}(\mathbf{x})$, and $\Delta_{h}=\frac{\partial h(\mathbf{x})}{\partial \mathbf{x}}\mathbf{m}(\mathbf{x})-\beta\left\|\frac{\partial h(\mathbf{x})}{\partial \mathbf{x}}\right\|\bm{\sigma}(\mathbf{x})$. Then the constraints given in \eqref{CLF_CBF_Condition} are satisfied with a probability of $(1-\delta)^{n}$.
\end{Lem}
\begin{proof}
Based on the condition~\eqref{Bound_Condition} provided in Lemma~\ref{Bound_Learning_Error},  we can infer that $\mathbb{P}\left\{\Delta_{V}\leq \frac{\partial V(\mathbf{x})}{\partial \mathbf{x}}\bm{\omega}(\mathbf{x})\right\}\geq (1-\delta)^{n}$ irrespective of whether  $\frac{\partial V(\mathbf{x})}{\partial \mathbf{x}}\geq 0$ or $\frac{\partial V(\mathbf{x})}{\partial \mathbf{x}}<0$. Consequently, we can ensure that $a(\mathbf{x})+\mathbf{b}(\mathbf{x})\mathbf{u}+\frac{\partial V(\mathbf{x})}{\partial \mathbf{x}}\bm{\omega}(\mathbf{x})\leq  -\kappa\zeta(\mathbf{x})$ with a probability $(1-\delta)^{n}$ if $\tilde{a}(\mathbf{x})+\mathbf{b}(\mathbf{x})\mathbf{u}\leq  -\kappa\zeta(\mathbf{x})$ is satisfied. Employing a similar proof process, it can be established that the CBF condition specified in~\eqref{CLF_CBF_Condition} also holds with a probability of $(1-\delta)^{n}$ when $\tilde{c}(\mathbf{x})+\mathbf{d}(\mathbf{x})\mathbf{u}\geq \rho\Gamma(\mathbf{x})$. 
\end{proof}
As the conditions given in~\eqref{CLF_CBF_GP_Condition} for ensuring safety and stability have been modified (compared to~\eqref{CLF_Condition} and~\eqref{CBF_Condition}), it is necessary to reassess the compatibility of these conditions to determine if safe stabilization remains feasible. For this consideration, we have the following lemma.
\begin{Lem}\label{Compatible}
Assume that both $\mathbf{b}(\mathbf{x})$ and $\mathbf{d}(\mathbf{x})$ are non-zero vectors, and the compatibility of the CLF $V(\mathbf{x})$ and CBF $h(\mathbf{x})$ w.r.t. the model~\eqref{Affine_Control_System} is established. The compatibility of the CLF $V(\mathbf{x})$ and CBF $h(\mathbf{x})$ for the model~\eqref{Affine_Control_System_Unknown} can be achieved with a probability of $(1-\delta)^{n}$ if the following conditions are satisfied 
\begin{equation}\label{CLF_compatible}
\begin{split}
    &\Delta_{V}\mathbf{d}(\mathbf{x})^{\top} \mathbf{b}(\mathbf{x})-\Delta_{h}\mathbf{b}(\mathbf{x})^{\top} \mathbf{b}(\mathbf{x})\geq 0.
\end{split}
\end{equation}
\begin{equation}\label{CBF_compatible}
      \Delta_{V}\mathbf{d}(\mathbf{x})^{\top} \mathbf{d}(\mathbf{x})-\Delta_{h}\mathbf{b}(\mathbf{x})^{\top} \mathbf{d}(\mathbf{x})\geq 0.
\end{equation}
\end{Lem}
\begin{proof}
    Since the CLF $V(\mathbf{x})$ and CBF $h(\mathbf{x})$ w.r.t. the model~\eqref{Affine_Control_System} is compatible, we know that one of the conditions provided in~\eqref{Compatibility_conditions} has to be satisfied. For the case that $\frac{\|\mathbf{b}(\mathbf{x})\mathbf{d}(\mathbf{x})^{\top}\|}{\|\mathbf{b}(\mathbf{x})\|\|\mathbf{d}(\mathbf{x})\|}= 1\,\,\, \text{and}\,\,\, v(\mathbf{x})\geq 0$, we will obtain $
        \tilde{v}(\mathbf{x})=v(\mathbf{x})+\Delta_{V}\mathbf{d}(\mathbf{x})^{\top} \mathbf{b}(\mathbf{x})-\Delta_{h}\mathbf{b}(\mathbf{x})^{\top} \mathbf{b}(\mathbf{x})\geq 0$ if~\eqref{CLF_compatible} is satisfied. Similarly, we have $\tilde{w}(\mathbf{x})=w(\mathbf{x})+\Delta_{V}\mathbf{d}(\mathbf{x})^{\top} \mathbf{d}(\mathbf{x})-\Delta_{h}\mathbf{b}(\mathbf{x})^{\top} \mathbf{d}(\mathbf{x})\geq 0$ if $\frac{\|\mathbf{b}(\mathbf{x})\mathbf{d}(\mathbf{x})^{\top}\|}{\|\mathbf{b}(\mathbf{x})\|\|\mathbf{d}(\mathbf{x})\|}= 1\,\,\, \text{and}\,\,\, w(\mathbf{x})\geq 0$. Since one of the following conditions can be satisfied, we can ensure the two conditions given in~\eqref{CLF_CBF_GP_Condition} are compatible.
\begin{equation}\label{Compatibility_updated_conditions}
\begin{split}
\begin{cases}
     &\frac{\|\mathbf{b}(\mathbf{x})\mathbf{d}(\mathbf{x})^{\top}\|}{\|\mathbf{b}(\mathbf{x})\|\|\mathbf{d}(\mathbf{x})\|}\neq 1,\\
     &\frac{\|\mathbf{b}(\mathbf{x})\mathbf{d}(\mathbf{x})^{\top}\|}{\|\mathbf{b}(\mathbf{x})\|\|\mathbf{d}(\mathbf{x})\|}= 1\,\,\, \text{and}\,\,\, \tilde{v}(\mathbf{x})\geq 0, \\
     &\frac{\|\mathbf{b}(\mathbf{x})\mathbf{d}(\mathbf{x})^{\top}\|}{\|\mathbf{b}(\mathbf{x})\|\|\mathbf{d}(\mathbf{x})\|}= 1\,\,\, \text{and}\,\,\, \tilde{w}(\mathbf{x})\geq 0.
     \end{cases}
\end{split}
\end{equation}
Next, with Lemma~\ref{GP_Sufficient}, the compatibility of the CLF $V(\mathbf{x})$ and CBF $h(\mathbf{x})$ for the model~\eqref{Affine_Control_System_Unknown} can be achieved with a probability $(1-\delta)^{n}$.
\end{proof}
\subsection{Compatible Case: Universal Formula with GPs}
In Lemma~\ref{Compatible}, we assume that the compatibility of the CLF $V(\mathbf{x})$ and CBF $h(\mathbf{x})$ w.r.t. the model~\eqref{Affine_Control_System} is primarily established. This is reasonable since, before employing GP regression, we usually first construct compatible CLF $V(\mathbf{x})$ and CBF $h(\mathbf{x})$ based on the model with unknown model uncertainties in~\eqref{Affine_Control_System}. Afterward, we use the conditions provided in~\eqref{CLF_compatible} and~\eqref{CBF_compatible} to guarantee the compatibility of  $V(\mathbf{x})$ and $h(\mathbf{x})$ for~\eqref{Affine_Control_System_Unknown}. However, for the case that no  prior information regarding the compatibility of the CLF $V(\mathbf{x})$ and CBF $h(\mathbf{x})$ concerning the model~\eqref{Affine_Control_System} is available, we can directly employ the condition~\eqref{Compatibility_updated_conditions} to guarantee the compatibility of the CLF $V(\mathbf{x})$ and CBF $h(\mathbf{x})$ for the model~\eqref{Affine_Control_System_Unknown}.
\begin{Thm}\label{Universal_Formula}
  Suppose that i) both $\mathbf{b}(\mathbf{x})$ and $\mathbf{d}(\mathbf{x})$ are non-zero vectors; ii) the compatibility of the CLF $V(\mathbf{x})$ and CBF $h(\mathbf{x})$ w.r.t. the model~\eqref{Affine_Control_System} is established; and iii) the conditions~\eqref{CLF_compatible} and ~\eqref{CBF_compatible} are satisfied. We design a universal formula for system~\eqref{Affine_Control_System_Unknown} as follows
  \begin{equation}\label{QP_GP_Control_Law_relaxed}
    \begin{split}
      \tilde{\mathbf{u}}_{\mathrm{Uni}}^{\star}(\mathbf{x})=\left\{\begin{array}{ll}
\tilde{\mathbf{m}}_{\mathrm{Uni-CLF}}^{\star}(\mathbf{x}),&\quad\mathbf{x}\in\tilde{\mathcal{P}}_{1}\\
\tilde{\mathbf{n}}_{\mathrm{Uni-CBF}}^{\star}(\mathbf{x}),&\quad\mathbf{x}\in\tilde{\mathcal{P}}_{2}\\
\tilde{\mathbf{p}}_{\mathrm{Uni}}^{\star}(\mathbf{x}),&\quad\mathbf{x}\in\tilde{\mathcal{P}}_{3}\\
\mathbf{0},&\quad\mathbf{x}\in\tilde{\mathcal{P}}_{4},
\end{array}\right.  
    \end{split}
\end{equation}
where $\tilde{\mathbf{m}}(\mathbf{x})$, $\tilde{\mathbf{n}}(\mathbf{x})$, $\tilde{\mathbf{p}}_{\mathrm{Uni}}^{\star}(\mathbf{x})$, $\tilde{\mathcal{P}}_{1}$, $\tilde{\mathcal{P}}_{2}$, $\tilde{\mathcal{P}}_{3}$, and $\tilde{\mathcal{P}}_{4}$ can be obtained from ~\eqref{QP_Control_Law_relaxed} by replacing $a(\mathbf{x})$ and $c(\mathbf{x})$ with $\tilde{a}(\mathbf{x})$ and $\tilde{c}(\mathbf{x})$. The safety and closed-loop stability for can be guaranteed with a probability $(1-\delta)^{n}$. 
\end{Thm}
\begin{proof}
Under the assumptions i), ii), and iii), we can establish the compatibility of the CLF $V(\mathbf{x})$ and CBF $h(\mathbf{x})$ for the model~\eqref{Affine_Control_System_Unknown} using Lemma~\ref{Compatible}, ensuring this compatibility with a probability of $(1-\delta)^{n}$. Followed by the application of a designed universal formula (by replacing $a(\mathbf{x})$ and $c(\mathbf{x})$ with $\tilde{a}(\mathbf{x})$ and $\tilde{c}(\mathbf{x})$ in ~\eqref{QP_Control_Law_relaxed}), which guarantees the satisfaction of condition~\eqref{CLF_CBF_GP_Condition}. As dictated by Lemma~\ref{GP_Sufficient}, it ensures that the constraints provided in \eqref{CLF_CBF_Condition} are satisfied with a probability of $(1-\delta)^{n}$. Consequently, the safety and closed-loop stability of the system~\eqref{Affine_Control_System_Unknown} is established with a probability of $(1-\delta)^{n}$.
\end{proof}
\subsection{Incompatible Case: Universal Formula with GPs}
In practice, the compatibility assumption, i.e., the condition \eqref{Compatibility_updated_conditions},  may not always hold. In the case of encountering incompatible safe stabilization, we need to consider relaxation strategies, where a typical choice is to sacrifice the stability requirement while taking safety constraints as a hard constraint. Following the method described in~\cite{universal_formula}, a universal formula accounts for incompatibility and can be designed accordingly. 
\begin{Thm}~\label{Universal_Formula_Relaxed}
    Assume that i) both $\mathbf{b}(\mathbf{x})$ and $\mathbf{d}(\mathbf{x})$ are non-zero vectors and ii) the compatibility of the CLF $V(\mathbf{x})$ and CBF $h(\mathbf{x})$ w.r.t. the model~\eqref{Affine_Control_System_Unknown} is not satisfied at a certain $\mathbf{x}$. The following universal formula allows us to ensure strict safety with a probability $(1-\delta)^{n}$ but no strict stability guarantees.
    \begin{equation}\label{QP_Control_Law_relaxed_first}
    \begin{split}
      \bar{\mathbf{u}}_{\mathrm{Uni}}^{\star}(\mathbf{x})=\left\{\begin{array}{ll}
\bar{\mathbf{m}}_{\mathrm{Uni-CLF}}^{\star}(\mathbf{x}),&\quad\mathbf{x}\in\bar{\mathcal{P}}_{1}\\
\bar{\mathbf{n}}_{\mathrm{Uni-CBF}}^{\star}(\mathbf{x}),&\quad\mathbf{x}\in\bar{\mathcal{P}}_{2}\\
\bar{\mathbf{p}}_{\mathrm{Uni}}^{\star}(\mathbf{x}),&\quad\mathbf{x}\in\bar{\mathcal{P}}_{3}\\
\mathbf{0},&\quad\mathbf{x}\in\bar{\mathcal{P}}_{4},
\end{array}\right.  
    \end{split}
\end{equation}
where $\bar{\mathbf{m}}_{\mathrm{CLF}}^{\star}(\mathbf{x})=-\frac{\tilde{a}(\mathbf{x})+\kappa\zeta(\mathbf{x})}{1/m+\mathbf{b}(\mathbf{x})\mathbf{b}(\mathbf{x})^{\top}} \mathbf{b}(\mathbf{x})^{\top}$, $\bar{\mathbf{n}}_{\mathrm{CBF}}^{\star}(\mathbf{x})= \frac{\rho\Gamma(\mathbf{x})-\tilde{c}(\mathbf{x})}{\mathbf{d}(\mathbf{x})\mathbf{d}(\mathbf{x})^{\top}} \mathbf{d}(\mathbf{x})^{\top}$,  $\bar{\mathbf{p}}_{\mathrm{QP}}=-\bar{\lambda}_{1}\bar{\mathbf{m}}_{\mathrm{CLF}}^{\star}(\mathbf{x})
+\bar{\lambda}_{2}\bar{\mathbf{n}}_{\mathrm{CBF}}^{\star}(\mathbf{x})$, $m>0$,
\begin{equation*}\label{bar_lambda_compute_no_ref}
    \begin{split}
        \begin{bmatrix}
\bar{\lambda}_{1}\\
\bar{\lambda}_{2}
\end{bmatrix}={\begin{bmatrix}
\frac{1}{m}+\mathbf{b}(\mathbf{x})\mathbf{b}(\mathbf{x})^{\top} & -\mathbf{b}(\mathbf{x})\mathbf{d}(\mathbf{x})^{\top}\\
-\mathbf{b}(\mathbf{x})\mathbf{d}(\mathbf{x})^{\top} & \mathbf{d}(\mathbf{x})\mathbf{d}(\mathbf{x})^{\top}
\end{bmatrix}}^{-1}\begin{bmatrix}
\tilde{a}(\mathbf{x})+\kappa\zeta(\mathbf{x})\\
\Gamma(\mathbf{x})-\rho \tilde{c}(\mathbf{x})
\end{bmatrix},
    \end{split}
\end{equation*}
\begin{equation*}\label{Domain_of_sets_relaxed_new}
\begin{aligned}
&\bar{\mathcal{P}}_{1}=\{\mathbf{x}\in\mathbb{R}^{n}|\tilde{a}(\mathbf{x})+\kappa\zeta(\mathbf{x})\geq 0,\bar{v}(\mathbf{x})< 0\},\\
&\bar{\mathcal{P}}_{2}=\{\mathbf{x}\in\mathbb{R}^{n}|\tilde{c}(\mathbf{x})-\rho\Gamma(\mathbf{x})\leq 0,\bar{w}(\mathbf{x})< 0\},\\
&\bar{\mathcal{P}}_{3}=\{\mathbf{x}\in\mathbb{R}^{n}|\bar{w}(\mathbf{x})\geq 0,\bar{v}(\mathbf{x})\geq 0, \\
&\mathbf{b}(\mathbf{x})^{\top} \mathbf{b}(\mathbf{x}) \mathbf{d}(\mathbf{x})^{\top} \mathbf{d}(\mathbf{x})-\mathbf{b}(\mathbf{x})^{\top} \mathbf{d}(\mathbf{x}) \mathbf{b}(\mathbf{x})^{\top} \mathbf{d}(\mathbf{x}) \neq 0\},\\
&\bar{\mathcal{P}}_{4}=\{\mathbf{x}\in\mathbb{R}^{n}|\tilde{a}(\mathbf{x})+\kappa\zeta(\mathbf{x})< 0,\tilde{c}(\mathbf{x})-\rho\Gamma(\mathbf{x})> 0\}.
\end{aligned}
\end{equation*}
where $\bar{w}(\mathbf{x})=\tilde{w}(\mathbf{x})$ and
$\bar{v}(\mathbf{x})=(\tilde{a}(\mathbf{x})+\kappa\zeta(\mathbf{x}))\mathbf{d}(\mathbf{x})\mathbf{d}(\mathbf{x})^{\top}-(\tilde{c}(\mathbf{x})-\rho\Gamma(\mathbf{x}))(\frac{1}{m}+\mathbf{b}(\mathbf{x}) \mathbf{d}(\mathbf{x})^{\top})$.
\end{Thm}
\begin{proof}
    As revealed in~\cite{universal_formula}, the control law presented in~\eqref{QP_Control_Law_relaxed_first} is an analytical solution for the following optimization problem.
\begin{equation}\label{CLF_CBF_QP_Proof}
\begin{split}
&\min_{\mathbf{u},\chi}\frac{1}{2}\left(\|\mathbf{u}\|^{2}+1/m\chi^{2}\right)\\
&\text{s.t.}\quad\tilde{a}(\mathbf{x})+\mathbf{b}(\mathbf{x})\mathbf{u}\leq-\kappa\zeta(\mathbf{x})+\chi,\\
&\qquad\tilde{c}(\mathbf{x})+\mathbf{d}(\mathbf{x})\mathbf{u}\geq \rho\Gamma(\mathbf{x}),
\end{split}
\end{equation}
By defining $\underline{\mathbf{u}}=[\mathbf{u},\chi]^{\top}$, the two inequality condition in~\eqref{CLF_CBF_QP_Proof} can be rewritten as:
\begin{equation}\label{Rewritten}
\begin{split}
\tilde{a}(\mathbf{x})+\underline{\mathbf{b}}(\mathbf{x})\underline{\mathbf{u}}\leq -\kappa\zeta(\mathbf{x}),\quad\tilde{c}(\mathbf{x})+\underline{\mathbf{d}}(\mathbf{x})\underline{\mathbf{u}}\geq \rho\Gamma(\mathbf{x}),
\end{split}
\end{equation}
where $\underline{\mathbf{b}}(\mathbf{x})=[\mathbf{b}(\mathbf{x}), 1]$ and $\underline{\mathbf{d}}(\mathbf{x})=[\mathbf{d}(\mathbf{x}), 0]$.
In this case, the compatible condition $\frac{\|\underline{\mathbf{b}}(\mathbf{x})\underline{\mathbf{d}}(\mathbf{x})^{\top}\|}{\|\underline{\mathbf{b}}(\mathbf{x})\|\|\underline{\mathbf{d}}(\mathbf{x})\|}\neq 1$ is always satisfied. with Lemma~\ref{GP_Sufficient}, the compatibility of the CLF $V(\mathbf{x})$ and CBF $h(\mathbf{x})$ for the model~\eqref{Affine_Control_System_Unknown} (with an introduction of the slacking variable $\chi$) can be achieved with a probability $(1-\delta)^{n}$. Next, using the universal formula from~\eqref{QP_Control_Law_relaxed_first}, we achieve $\tilde{c}(\mathbf{x})+\mathbf{d}(\mathbf{x})\mathbf{u}\geq 0$, but we cannot ensure that $\tilde{a}(\mathbf{x})+\mathbf{b}(\mathbf{x})\mathbf{u}\leq 0$. Based on Lemma~\ref{GP_Sufficient}, we achieve strict safety with a probability $(1-\delta)^{n}$. However, since the stability is relaxed as shown in~\eqref{CLF_CBF_QP_Proof}, the control law~\eqref{QP_Control_Law_relaxed_first} does not provide stability guarantees.
\end{proof}
\subsection{Deal with Control Limit: A Projection-based Approach}
    For both Theorem~\ref{Universal_Formula} and Theorem~\ref{Universal_Formula_Relaxed}, we did not consider the control limits of $\mathbf{u}$. However, this assumption is not practical, as an actual system usually has a control limit. To address this problem, we first provide the following assumption.
    \begin{equation}
     \tilde{\mathcal{U}}=\{\mathbf{u}\in\mathbb{R}^{m}|\mathbf{u}_{\mathrm{min}}\leq \mathbf{u}\leq \mathbf{u}_{\mathrm{max}}\},
\end{equation}
where $\mathbf{u}_{\mathrm{min}}$ and $\mathbf{u}_{\mathrm{max}}$ are constant vectors, which denote the minimal and maximal values of the control input.  Next, we perform a projection operation to the universal formula that aligns it with the nearest point within the set $\tilde{\mathcal{U}}$. To clarify this concept, we introduce the following definitions: $\mathcal{S}_{V}(\mathbf{x})=\{\mathbf{x}\in\mathbb{R}^{n}|\tilde{a}(\mathbf{x})+\mathbf{b}(\mathbf{x})\mathbf{u}\leq -\kappa\sigma(\mathbf{x})\}$ and $\mathcal{S}_{h}(\mathbf{x})=\{\mathbf{x}\in\mathbb{R}^{n}|\tilde{c}(\mathbf{x})+\mathbf{d}(\mathbf{x})\mathbf{u}\geq \rho\Gamma(\mathbf{x})\}$. Afterward, we provide a graphical interpretation in Fig.~\ref{Projection_Inter}~\footnotetext[1]{To facilitate graphical interpretation, we have chosen to set the control input dimension to 2. However, this projection idea can be extended to cases of arbitrary dimensions.}. By denoting the projected control input as $\mathbf{u}_{\mathrm{proj}}$, the optimal control input can be obtained by solving the following optimization problem.
\begin{equation}\label{QP_projection}
        \begin{split}
            \min\limits_{\mathbf{u}_{\mathrm{proj}}\in\mathbb{R}^{m}}&\frac{1}{2}\|\mathbf{u}_{\mathrm{proj}}-\mathbf{u}_{\mathrm{Uni}}\|^{2}\\
            &\mathbf{u}_{\mathrm{min}}\leq \mathbf{u}_{\mathrm{Proj}}\leq \mathbf{u}_{\mathrm{max}}.
        \end{split}
    \end{equation}
\begin{figure}[tp]
 \centering
    \makebox[0pt]{%
    \includegraphics[width=2.5in]{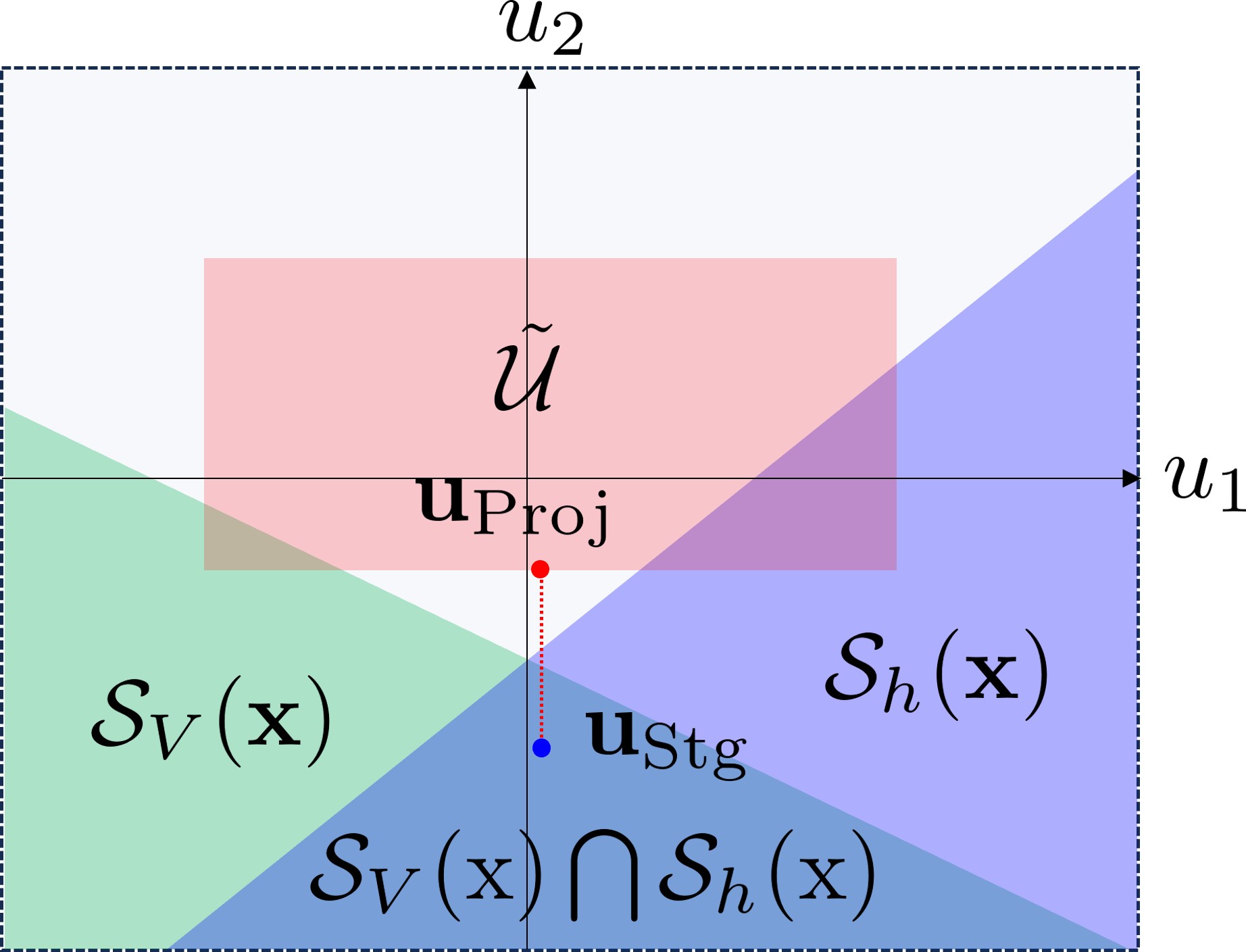}}
    \caption{A graphical interpretation of the projection idea.}
    \label{Projection_Inter}
\end{figure}

\noindent
As shown in Fig.~\ref{Projection_Inter}, the control input $\mathbf{u}_{\mathrm{Proj}}$ falls within the admissible range of control inputs, denoted as $\tilde{\mathcal{U}}$. It is true that $\mathbf{u}_{\mathrm{Proj}}$ lacks stability and safety guarantees, as it does not belong to either $\mathcal{S}_{V}(\mathbf{x})$ or $\mathcal{S}_{h}(\mathbf{x})$. However, this operation can be interpreted as a relaxation of both stability and safety requirements but requires a mandatory satisfaction of the control input limit. This relaxation is acceptable under certain circumstances, particularly when strict stability guarantees are not a necessity, and there are some safety margins.  
\section{Application Study and Simulation Results}\label{Application_study}
\begin{figure*}[tp]
 \centering
    \makebox[0pt]{%
    \includegraphics[width=7.5in]{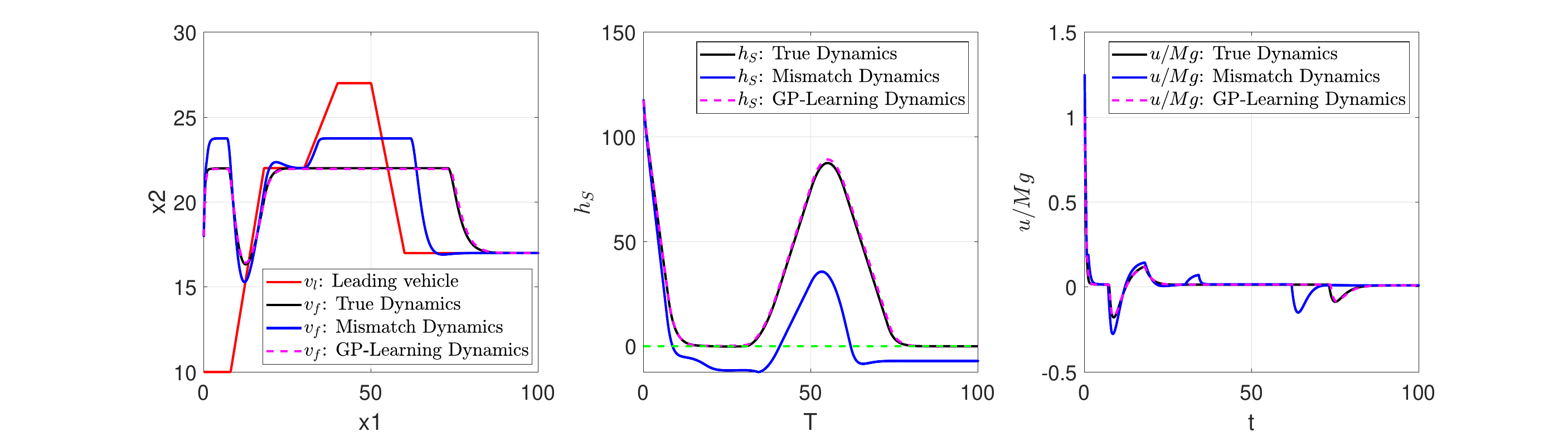}}
    \caption{Simulation results of the ACC problem for three different scenarios: universal formula applied to the true system dynamics (serving as a benchmark), mismatched dynamics, and GP learning-based dynamics. Left: the behavior of the follower (controlled) vehicle. Middle: the value of the CBF. Right: the behavior os the control input (scaled by $1/Mg$).
    }
    \label{Simulation_Results}
\end{figure*}
In this section, we evaluate the performance of our GP-based universal formula using an ACC example. The simulation results showcase the efficiency of the proposed solution. 

The application scenario involves two vehicles, the lead and the follower, both traveling in a straight line. The follower vehicle is equipped with ACC, and its mission is to attain the desired speed while ensuring a safe distance from the lead vehicle. This situation can be succinctly described by the compact expression of the system dynamics:
\begin{equation}\label{ACC_model}
\dot{\mathbf{x}}=\left(\begin{array}{c}
-F_{r} / M \\
a_{L} \\
x_{2}-x_{1}
\end{array}\right)+\left(\begin{array}{c}
1 / M \\
0 \\
0
\end{array}\right) \mu ,
\end{equation}
where $\mathbf{x}=\left(x_{1},x_{2},x_{3}\right):=\left(v_{f},v_{l},D\right)$, $v_{f}$ and $v_{l}$ are the velocities of the following and leading vehicle, $D$ is the distance between the two vehicles, $M$ is the mass of the follower vehicle; $F_{r}(x)=f_{0}+f_{1}v_{f}+f_{2}v_{f}^{2}$ is the aerodynamic drag with constants $f_{0}$, $f_{1}$ and $f_{2}$ determined empirically; $a_{L}\in\left[-a_{l}g,a_{l}^{\prime}g\right]$ is the overall acceleration/deceleration of the lead vehicle with $a_{l}$, $a_{l}^{\prime}$ fractions of the gravitational constant $g$ for deceleration and acceleration, respectively. The control input $\mu\in\mathbb{R}$ of the following car is wheel force. The quantities of the parameters of the true model are listed in TABLE I. 

To incorporate model uncertainty, we deliberately deviate from the nominal model by employing different parameters: $f_{0}=10$, $f_{1}=5$, and $f_{2}=0.25$, rather than the parameters specified in TABLE I. Moreover, the safety requirement for the system described in Equation~\eqref{ACC_model}, which indicates ``maintaining a safe distance from the lead vehicle". To satisfy this criterion, we enforce the constraint that the minimum distance between the two vehicles should be equivalent to ``half the speedometer"~\cite{Hard_Constraint}. To reach the desired velocity while preserving a safe distance, we introduce the CLF and CBF for the ACC system defined in Equation \eqref{ACC_model}, which are detailed below:
\begin{equation*}
    \begin{split}
    &V(\mathbf{x})=(x_{1}-v_{d})^{2},\quad h(\mathbf{x})=x_{3}-\tau_d x_{1},
    \end{split}
\end{equation*}
where $\tau_{d}$ and $v_d$ are given in TABLE I, which are both constants. 
\begin{table}[tp]
\footnotesize
\centering
\caption{Parameter Quantities Used in Simulation}
\begin{tabular}{cccccc}
  \hline
 $M$ & $1650\,\, \mathrm{kg}$ & $g$ & $9.81\,\,\mathrm{m/s^{2}}$ & $v_{l}(0)$ & $15\,\, \mathrm{m/s}$\\
 $f_{0}$ & $0.1\,\, \mathrm{N}$ & $v_{d}$ & $22\,\,\mathrm{m/s}$ & $D(0)$ & $150\,\,\mathrm{m}$\\
 $f_{1}$ & $5\,\, \mathrm{Ns/m}$ & $m$ & $100$ & $\tau_{d}$ & $1.8$\\
 $f_{2}$ & $0.25\,\, \mathrm{Ns^{2}/m}$ & $v_{f}(0)$ & $18\,\, \mathrm{m/s}$ & $a_{l}$ & $0.3$\\
  \hline
\end{tabular}
\label{table: index1}
\end{table}
By calculating the Lie derivative of $h(\mathbf{x})$ and $V(\mathbf{x})$ (setting $\lambda=4$ and $\beta=0.5$ in~\eqref{CLF_Condition} and~\eqref{CBF_Condition}), we have
$a(\mathbf{x}) = -2(x_{1}-v_{d})F_r/M+4V(\mathbf{x})$, $\mathbf{b}(\mathbf{x}) = [2(x_{1}-v_{d})F_r/M, 0, 0]$, $c(\mathbf{x}) = x_{2}-x_{1}+\tau_{d}F_r/M+0.5h(\mathbf{x})$, and
$\mathbf{d}(\mathbf{x}) = [-\tau_{d}/M, 0, 0]$.
In the simulation, we specifically consider a scenario where the assumption of compatibility between the CLF and the CBF is not met. In particular, we provide the value of $a_{L}$ as follows:
\begin{equation*}\label{Acceleration}
    \begin{split}
      a_{L}=\left\{\begin{array}{ll}
0\,\,\mathrm{m/s^{2}} &t\leq 8\,\,\mathrm{s}\\
1.2\,\,\mathrm{m/s^{2}} &8<t\leq 18\,\,\mathrm{s}\\
0\,\,\mathrm{m/s^{2}} &18<t\leq 30\,\,\mathrm{s}\\
0.5\,\,\mathrm{m/s^{2}} &30<t\leq 40\,\,\mathrm{s}\\
0\,\,\mathrm{m/s^{2}} &40<t\leq 50\,\,\mathrm{s}\\
-1\,\,\mathrm{m/s^{2}} &50<t\leq 60\,\,\mathrm{s}\\
 0\,\,\mathrm{m/s^{2}}&\mathrm{otherwise}
\end{array}\right.  
    \end{split}
\end{equation*}
To construct the universal formula as in~\eqref{QP_Control_Law_relaxed_first}, we set $m=10$, $\kappa=0.2$, and $\rho=0.1$. Next, we are devoted to evaluating the efficacy of the GP-based universal formula. To gather data for GP regression, we conducted initial experiments as follows: We randomly selected an initial state from the distribution $\mathcal{X}_{0}$ and allowed the system to evolve over a finite time interval using various control inputs. These inputs covered the entire range of control input scenarios relevant to the safe stabilization problem. Using a sample-based approach, we collected discrete-time state and input history data. Subsequently, we derived the measurement $\mathbf{y}^{(q)}$ for the unknown function $\bm{\omega}(\mathbf{x})$, as outlined in Assumption~\ref{Data_Collection}.
\begin{equation}
\mathbf{y}^{(q)}=\dot{\mathbf{x}}^{(q)}-\mathbf{f}(\mathbf{x}^{(q)})-\mathbf{g}(\mathbf{x}^{(q)})\mathbf{u}^{(q)}+\bm{\varepsilon}^{(q)}, q=1,\cdots Q.
\end{equation}
where $\dot{\mathbf{x}}^{(q)}$ is assumed to be known, $Q=2000$ is the size of the training dataset. By using the GPML toolbox~\cite{GPML}, we derive the mean function and covariance function as specified in Equations~\eqref{mean_function} and~\eqref{covariance_function}. Subsequently, due to the incompatibility of the safe stabilization problem under consideration, the universal formula following Equation~\eqref{QP_Control_Law_relaxed_first} is constructed.

Fig.~\ref{Simulation_Results} presents the simulation results including three different scenarios: universal formula applied to the true system dynamics (serving as a benchmark), mismatched dynamics, and GP learning-based dynamics. In the left side of Fig.~\ref{Simulation_Results}, we observe that the safe stabilization task is successfully achieved by using the universal formula with GP, where the controlled vehicle fails to attain the desired cruising speed (at all times) but without violating the distance constraint. This aligns with the conclusion from Theorem~\ref{Universal_Formula_Relaxed}, indicating that the universal control law lacks closed-loop stability but has strict safety guarantees. Besides, the GP-based universal formula's performance is nearly identical to the scenario with true dynamics, primarily due to the successful GP learning, achieved through the selection of $\delta=0.987$ in~\eqref{Bound_Condition}. However, as indicated in both the left and middle figures of Fig.~\ref{Simulation_Results}, the universal formula with mismatched dynamics fails to ensure safety as $h_{S}<0$, which strengthens the effectiveness of the GP-based universal formula. Finally, as depicted on the right side of Fig.~\ref{Simulation_Results}, the control input, scaled by $1/Mg$, is provided to show the behavior of the universal control law as outlined in Equation~\eqref{QP_Control_Law_relaxed_first}.
\section{Conclusions}
In this paper, a novel approach that combines GPs with the universal formula is proposed to address safe stabilization under model uncertainties. Our results demonstrate that this approach offers high probability of closed-loop stability and safety guarantees, even when the model is inaccurate. We also introduced compatibility conditions to assess the conflicts between GP-based CLF and CBF conditions, which determines whether a universal formula can be constructed to address a safe stabilization problem. Furthermore, in scenarios where compatibility assumptions do not hold and control system limits are present, we presented a modified universal formula that relaxes stability constraints and a projection-based method that considers control limits. The effectiveness of our approach was validated through simulations, with a focus on the ACC application. These results highlight the potential practical applications of our method in real-world scenarios, offering a promising avenue for addressing safe stabilization challenges in the presence of model uncertainties.
\bibliographystyle{IEEEtran}	
\bibliography{Univ_GP}
\end{document}